\documentclass[12pt]{amsart}
\usepackage{amssymb,verbatim,enumerate,ifthen}
\usepackage{cite}
\usepackage[mathscr]{eucal}
\usepackage[utf8]{inputenc}
\usepackage[T1]{fontenc}
\usepackage{esint}
\usepackage{marginnote}
\usepackage{xcolor}
\usepackage{tikz}
\usetikzlibrary{automata,positioning,arrows.meta,shapes}

\newtheorem{thm}{Theorem}[section]
\newtheorem{obs}{Observation}
\newtheorem{cor}[thm]{Corollary}
\newtheorem{prop}[thm]{Proposition}
\newtheorem{lem}[thm]{Lemma}

\DeclareMathOperator{\rank}{rank}

\theoremstyle{definition}

\newtheorem{exa}[thm]{Example}

\numberwithin{equation}{section}
\def\eq#1{{\rm(\ref{#1})}}
\def\Eq#1#2{\ifthenelse{\equal{#1}{*}}
  {\begin{equation*}\begin{aligned}[]#2\end{aligned}\end{equation*}}
  {\begin{equation}\begin{aligned}[]\label{#1}#2\end{aligned}\end{equation}}}

\newcommand\id{\mathrm{id}}
\def\A{\mathscr{A}}

\def\P{\mathscr{P}}
\newcommand\dd{\mathbf{d}}

\def\root{\mathcal{R}}

\newcommand\MM{\mathbf{M}}
\newcommand\KK{\mathbf{K}}
\newcommand\R{\mathbb{R}}
\newcommand\N{\mathbb{N}}
\newcommand\Z{\mathbb{Z}}

\author{P\'al Burai}
\author{Pawe{\l} Pasteczka}
\title{Mathematical model of information bubbles on networks}
\DeclareMathOperator{\Succ}{succ}
\DeclareMathOperator{\source}{source}
\DeclareMathOperator{\Prec}{prec}
\keywords{Information spreading; invariant means;  mean-type mappings; networks; information bubbles; aggregation functions}
\subjclass[2020]{{\sc Primary:} 94C15;\\ {\sc Secondary:} 05C82, 26E60, 91D30}
\begin{document}
\begin{abstract}
    The main goal of this paper to introduce a new model of evolvement of narratives (common opinions, information bubble) on networks. Our main tools come from invariant mean theory and graph theory. The case, when the root set of the network (influencers, news agencies, etc.) is ergodic is fully discussed. The other possibility, when the root contains more than one component is partially discussed and it could be a motivation for further research.
\end{abstract}

\maketitle

\section{Introduction}

Our basic object is a directed graph, which is an abstract visualization of a network. A number is assigned to all vertices, which symbolize the information or the "opinion" of the vertex (users) in question. Essentially, we distinguish to types of vertices: roots (information sources or influencers) and the others. From a root, all the other non root vertices are available. So, the opinion of a root, or other words, the information, which is delivered by an influencer has impact on non root vertices. Roots can influenced only by other (not necessarily all of them) root elements. 

We investigate the following situation. Some roots maybe cooperate (possibly all), so they somehow form a common opinion at last. The mathematical model of this are multivariable means, which aggregates all the opinions of the cooperating roots. Every vertex has an own aggregating rule, which possibly different from other vertices' rules. The question is, what happens after some time, if they repeat the aggregation process. There will be a common opinion at last or not? The key tool is the existence of invariant means, and their uniqueness. This highly depends on the structure of the incidence graph of the aggregating functions of the roots, and on the structure of the set of the roots.

\subsection{Description of the model}
For the mathematical model of an information bubble in a network (taking shape of a narrative) our main tool will be the concept of invariant means. Assume that players are not randomly milling
around, but sedentary, that is a graph of mutual relations between them does not change in time. Then they interact only with their neighbors. In other words, in our model we have a given (fixed) social graph and the aggregating function (a mean) at each node. In each step (iteration) every individual (independently) modifies his opinion by averaging neighbours' opinions.  In general, averaging takes into account the opinion of neighbours (possibly also the opinion of the individual in the previous step) and aggregates them into the opinion in the next time period. These concepts are well known in the literature, see for example Dahlgren~\cite{Dah21}, Gupta--Mishra \cite{GupMis21}, Liu--Zhang \cite{LiuZha14}, Ohtsuki--Nowak \cite{OhtNow06}, Polanski \cite{Pol24}, Yoon \cite{Yoo11}, and Zhang--Fu--Chen--Xie--Wang \cite{ZhaFuChe15}.
However, most of the models assume implicitly that individuals possess a very limited number of opinions (strategies). Then each vertex of the social graph choose a corresponding opinion (strategy) from this limited set. 

In this paper, we would like to generalize this model in a few directions. First of all, we allow continuous set of opinions, for simplicity, we assume that they are taken from the subinterval $I$ of reals. Next, each individual has its own aggregating function, which will be a mean defined on an interval $I$. 

The meaningful mathematical background was given in \cite{Pas23b} which was our motivation. That approach was, regretfully, too limited to cover the social networks and spreading informations. It did not allow to model the classical situations like press agencies, influencers, media, etc. 

Our new results model a more general situation where we have roots of the information (which describe a sort of privileged position in the social graph) and peers, which are (indirectly) connected to (a few or all) roots via the social graph. The aggregation is isolated from the external data except for fixing the initial values. The (social) network is represented by a directed graph. The issue is to obtain the total agreement at the very end. Such a state is not expected in the real world, on the other hand, it is not very surprising when we take into account that no new information is delivered to the network and the only evolution of opinions is caused by aggregating the opinions of others. It turned out that such convergence of iterations (taking shape of a narrative) is naturally connected to the notion of invariant means.

We assume that roots are connected (possibly indirectly) to each other, which is clearly necessary to obtain the global agreement in the limit process. This assumption is not artificial since in most cases roots are aware one of another -- directly or not. 
Moreover, we make a purely technical assumption that the graph induced by roots is aperiodic (which is satisfied by nearly all realistic social networks). Indeed, it is quite natural to assume that the root of a social network is aperiodic, for example, it holds if it contains a totally connected subgraph with at least three vertexes.

\subsection{Outline of the paper}

The second section contains the needed definitions and tools. We introduce here the concepts of means, mean-type mappings, invariant means, and the root set of a directed graph, which is one of the key concepts of this paper. The characterization theorem of the root set is also there.

The third section contains the main results of this work. The case, when the root set is ergodic, is completely investigated.  The ergodic property of the root entails the existence of a unique invariant mean. This can be considered as a common narrative of the network. It is also proved here, that we can expect such a simple situation only when the root is ergodic. 

In the fourth section, we provide illustrations, including a significant special scenario where the averaging of all users and influencers involves weighted arithmetic means. In this case, discrete Markov chains (which are essentially matrices) can be employed to compute the limit, resulting in an invariant mean.

We close the paper with conclusions and the summary of possible further research directions.

\section{Preliminaries}
\subsection{Mathematical tools from graph theory and the concept of the root}

Now we recall some elementary facts concerning graphs. For details, we refer the reader to the classical book \cite{GraKnuPat89}.

A \emph{digraph} is a pair $G=(V,E)$, where $V$ is a finite (possibly empty) set of vertexes, and $E\subset V \times V$ is a set of edges. For each $v \in V$ we denote by $N_G^-(v)$ and $N_G^+(v)$ sets of \emph{in-neighbors} and \emph{out-neighbors}, respectively. More precisely $N_G^-(v)=\{w \in V \colon (w,v)\in E\}$ and $N_G^+(v)=\{w \in V \colon (v,w)\in E\}$.  The edges of the form $(v,v)$ for $v \in V$ are called \emph{loops}. Let us observe that in view of the above definition the null graph (empty graph) $\varnothing:=(\emptyset,\emptyset)$ a well-defined digraph. 

A sequence $(v_0,\dots,v_n)$ of elements in $V$ such that $(v_{i-1},v_{i})\in E$ for all $i \in \{1,\dots,n\}$ is called a \emph{walk} from $v_0$ to $v_n$. The number $n$ is a \emph{length} of the walk. If for $v, w \in V$ there exists a walk from $v$ to $w$ in $G$, then we denote it by $v \leadsto_G w$ (abbreviated to $v \leadsto w$ whenever $G$ is known). A graph $G$ is called \emph{irreducible} provided $v \leadsto w$ for all $v, w \in V$. 

A \emph{cycle} in a graph is a non-empty walk in which only the first and last vertices are equal. A directed graph is said to be \emph{aperiodic} if there is no integer $k > 1$ that divides the length of every cycle of the graph. A~graph which in nonempty, irreducible and aperiodic is called \emph{ergodic}.

A topological ordering of a digraph $G=(V,E)$ is a linear ordering of its vertices such that for every directed edge $(v,w) \in E$, $v$ comes before $w$ in the ordering. It is known that if $G$ has no cycles, then there exists its topological ordering (see, for example, \cite[section~22.4]{CorLeiRivSte09}). Obviously, it is not uniquely determined.

We also need a lemma which will be useful in the remaining part of this paper
\begin{lem}[\!\!\cite{Pas23b}, Lemma~1] \label{lem:walkLB}
Let $G=(V,E)$ be an ergodic digraph. Then there exists $q_0$ such that for all $q \ge q_0$, and $v,w \in V$ there exists a walk from $v$ to $w$ of length exactly $q$.
\end{lem}

Let us now introduce the decomposition of a directed graph into a directed acyclic graph of its strongly connected components (see, for example, \cite[section~22.5]{CorLeiRivSte09}). More precisely, for a directed graph \mbox{$G=(V,E)$} we define a relation $\sim$ on its vertexes in the following way: $v \sim w$ if and only if they are both in the same strongly connected component (that is $v=w$ or there is a walk from $v$ to $w$ and from $w$ to $v$).  Obviously $\sim$ is an equivalence relation on $V$, thus we define the quotient graph $G^{SCC}:=G/_{\sim}$.
In more details $G^{SCC}=(V^{SCC},E^{SCC})$, where $V^{SCC}=V/_\sim$ and
\begin{eqnarray*}
    E^{SCC}=\big\{ (P,Q)\in  V^{SCC}  \times V^{SCC} \colon P \ne Q \text{ and }(p,q) \in E\\ \text{ for some } p \in P\text{ and }q\in Q\big\}.
\end{eqnarray*}

It can be shown that $G^{SCC}$ has no cycles. Now we define the set of sources a directed graph $G=(V,E)$ as follows
\begin{eqnarray*}
    \source(G):=\{v \in V \colon \text{ there is no edge in }E\text{ which ends in }v\}\\
    =\{v \in V \colon N_G^-(v)=\emptyset\}.
\end{eqnarray*}

Obviously, there are no edges between elements in the source. Furthermore, since $G^{SCC}$ is acyclic, that is, it has no cycles in $G^{SCC}$, we know that $\source(G^{SCC})$ is nonempty. In fact, it contains the first element of (any) topological ordering of $G^{SCC}$ (see the definition above).
In the next step, we go backward (to the initial graph $G$) and define the \emph{root of $G$} by 
\begin{equation}\label{D:definition of the root set}
    R(G):=\bigcup \source (G^{SCC}) \subset V.
\end{equation}

\begin{exa}
Let 
\Eq{*}{
G&=\{V,E\},\qquad V=\{a,b,c,d,e,f\},\\
E&=\{(a,d),(d,a),(b,c),(c,b),(d,e),(b,e),(e,f),(c,f)\}.
}
Then 
\Eq{*}{
G^{SCC}&=\{V^{SCC},E^{SCC}\},\qquad V^{SCC}=\{P,Q,R,S\},\\
E^{SCC}&=\{(P,R),(Q,R),(Q,S),(R,S)\},
}
where the equivalence classes $P,Q,R,S$ correspond to the sets $\{a,d\},$ $\{b,c\},\{e\},\{f\}$
respectively (see Figure \ref{F:exampleforroot}).

So, the source of $G^{SCC}=\{P,Q\}$, which entails that the root of $G$ is 
\[
R(G)=\{a,b,c,d\}.
\]
\begin{figure}[h]
\centering
\begin{tikzpicture}[>={Stealth[width=6pt,length=9pt]}, node distance=15mm,auto]
\filldraw [fill=black!10,draw=white] (-0.5,-0.5) rectangle (0.5,2.5);
\filldraw [fill=black!10,draw=white] (1.5,1.5) rectangle (4.5,2.5);
\begingroup
\newcommand{\dnode}{\node[state,inner sep=3pt,minimum size=1pt]}
\dnode (a) at (0,2) {$a$};
\dnode (b) at (2,2) {$b$};
\dnode (c) at (4,2) {$c$};
\dnode (d) at (0,0) {$d$};
\dnode (e) at (2,0) {$e$};
\dnode (f) at (4,0) {$f$};

\dnode (P) at (8,2) {$P$};
\dnode (Q) at (10,2) {$Q$};
\dnode (R) at (8,0) {$R$};
\dnode (S) at (10,0) {$S$};

\path[->] (a) edge [bend left] node {} (d);
\path[->] (d) edge [bend left] node {} (a);
\path[->] (b) edge [bend left] node {} (c);
\path[->] (c) edge [bend left] node {} (b);
\path[->] (d) edge node {} (e);
\path[->] (b) edge node {} (e);
\path[->] (c) edge node {} (f);
\path[->] (e) edge node {} (f);

\path[->] (P) edge node {} (R);
\path[->] (Q) edge node {} (R);
\path[->] (Q) edge node {} (S);
\path[->] (R) edge node {} (S);
\endgroup

\end{tikzpicture}
\caption{A directed graph $G$ and the corresponding $G^{SCC}$}
\label{F:exampleforroot} 
\end{figure}
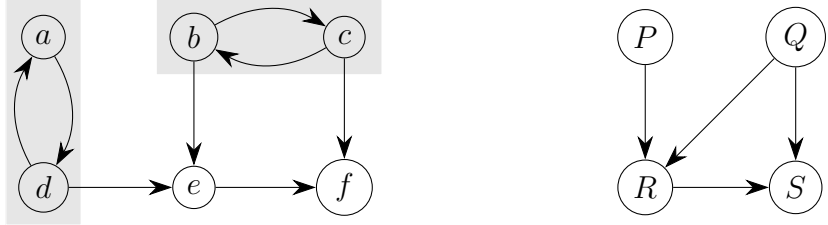
Here the root is not ergodic.
\end{exa}

\begin{exa}\label{ex2}
Let 
\Eq{*}{
G&=\{V,E\},\qquad V=\{1,2,3,4\},\\
E&=\{(1,1),(2,2),(4,4),(1,2),(2,1),(2,3),(3,4),(4,3)\}.
}
Then 
\[
G^{SCC}=\{V^{SCC},E^{SCC}\},\qquad V^{SCC}=\{I,II\},\quad E^{SCC}=\{(I,II)\},
\]
where the equivalence classes $I,II$ correspond to the sets $\{1,2\},\{3,4\}$ respectively (see Figure~\ref{fig2}).

So, the source of $G^{SCC}$ equals to $\{I\}$, which entails that the root of $G$ is 
\[
R(G)=\{1,2\}.
\]
Here the root is ergodic.
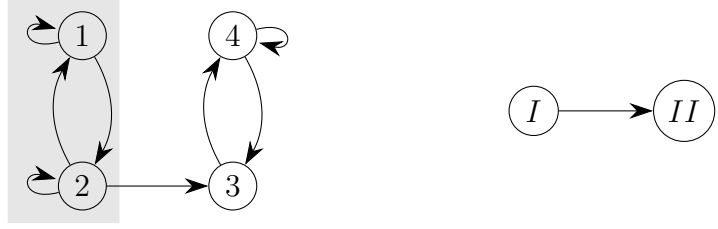
\begin{figure}[h]
\centering
\begin{tikzpicture}[>={Stealth[width=6pt,length=9pt]},
node distance=15mm,auto]
\filldraw [fill=black!10,draw=white] (-1,-0.5) rectangle (0.5,2.5);
\node[state,inner sep=3pt,minimum size=1pt] (b) at (0,0) {$2$};
\node[state,inner sep=3pt,minimum size=1pt] (c) at (2,0) {$3$};
\node[state,inner sep=3pt,minimum size=1pt] (d) at (2,2) {$4$};
\node[state,inner sep=3pt,minimum size=1pt] (a) at (0,2) {$1$};
\path[->] (a) edge [bend left] node {} (b);
\path[->] (a) edge [loop left] node {} (a);
\path[->] (b) edge [bend left] node {} (a);
\path[->] (b) edge [loop left] node {} (b);
\path[->] (b) edge node {} (c);
\path[->] (c) edge [bend left] node {} (d);
\path[->] (d) edge [bend left] node {} (c);
\path[->] (d) edge [loop right] node {} (d);

\node[state,inner sep=3pt,minimum size=1pt] (I) at (6,1) {$I$};
\node[state,inner sep=3pt,minimum size=1pt] (II) at (8,1) {$II$};
\path[->] (I) edge node {} (II);
\end{tikzpicture}
\caption{Graph $G_\alpha$ related to Example~\ref{ex2}.}
\label{fig2}\end{figure}

%
%
%
\end{exa}
Now we show the equivalent definition of the root.
\begin{thm}[Characterization theorem of $R(G)$]\label{T:characterization_of_the_root}
Let $G=(V,E)$ be a directed graph. Then $R(G)$ is the smallest subset $S \subset V$ such that the following conditions are valid:
\begin{enumerate}[(i)]
    \item For all $v\in V$ such that there exists a path $w \leadsto_G v$ for some $w \in S$;
    \item if $v \in S$ and there is an edge $(wv)\in E$ for some $w \in V$, then $w \in S$.
\end{enumerate}
\end{thm}
\begin{proof}
This proof is split into three steps.

\textsc{Step 1.} 
 We prove that the root set is a set of minimal elements with respect to a certain ordering in V.

Let us introduce the relation on $V$ as follows 
\Eq{*}{
p \prec q :\iff p\ne q,\ p \leadsto_G q \text{ and }q \not \leadsto_G p.
}

First, observe that $\prec$ is a strict partial ordering of elements in $V$. 

Second, if $p$ and $q$ are in the same SCC then $p \not \prec q$. Whence each element of $R(G)$ is minimal with respect to $\prec$ in $V$. 

Conversely, if $p$ is not a minimal element (with respect to $\prec$) then there exists a $\bar p$ such that $\bar p \prec p$. Then $\bar p$ is in the different SCC than $p$ and there exists a path $\bar p \leadsto p$. Consequently $[p]_\sim \notin \source(G^{SCC})$, and whence $p \notin R(G)$.  That is, $R(G)$ is exactly the set of minimal elements in the ordering $\prec$. 

\textsc{Step 2.} 
We show that conditions (i) and (ii) holds for $S=R(G)$.

Indeed, for every $v \in V$ there exists $w \in R(G)$ such $w \prec v$, which implies that condition (i) holds for $S=R(G)$.

Now assume that $v \in R(G)$ and $(wv) \in E$ for some $w \in G \setminus\{v\}$. Then, since $v$ is minimal we have $w \not \prec v$. Whence one of three cases holds: $w=v$ (which we can exclude), $w \not \leadsto_G v$ (which is impossible since $(wv) \in E$) or $v \leadsto_G w$ which implies that $v$ and $w$ are in the same SCC. Whence $w \in R(G)$, which implies that (ii) holds for $S=R(G)$.

\textsc{Step 3.} 
Now take any set $S \subset V$ such that condition (i) and (ii) holds. Observe that if $v \in S$ and $w \prec v$ then $w \leadsto_G v$ and (applying condition (ii) inductively) we get $w \in S$. 

Now take any $v \in R(G)$. Applying condition (i), there exists $v^* \in S$ such that $v^* \leadsto_G v$. Then, since $v^* \not\prec v$ we have that $v=v^*$ or  $v \leadsto_G v^*$, and therefore $v \in S$. Thus $R(G) \subseteq S$.
\end{proof}

Now, we define a \emph{root graph} $\root(G)$ as the graph induced by the root of $G$. Thus, purely formally, $\root(G):=(R(G),E\cap (R(G) \times R(G))$.

We underline a few easy observations related to this definition. 
\begin{obs}
\begin{enumerate}
\item Since $\source(G^{SCC})$ is nonempty, we get that $\root(G)$ is nonempty if $V$ is nonempty.
\item A graph $G$ is irreducible if, and only if, all its vertexes belong to the root, that is $\root(G)=G$.
\item There are no edges in $G$ that start outside the root and end inside it. 
\item $R(G)$ is a union of irreducible graphs. Consequently, $R(G)$ is irreducible if, and only if, $\source(G^{SCC})$ is a singleton.
\end{enumerate}
\end{obs}

A subset of the roots for which the generated graph is irreducible is called a \emph{component} of the root. 

The simplest situation is when the root graph is ergodic. In particular, there is only one component of the root set. We will see later, that in this case the effect of the common opinion of the root elements, which exists in this case,  will be the common opinion of the whole network (see Theorem \ref{thm:main}).

 The more challenging case is, when there is more than one components of $R(G)$. This can happen when there is no full cooperation in the root set. This issue is illustrated in Example~\ref{exa:E1}.

\subsection{Means, mean-type mappings and invariant means} Before we proceed further recall that, for a given $p\in \N$ and an interval $I \subset \R$, a \emph{$p$-variable mean on $I$} is an arbitrary function $M \colon I^p \to I$ satisfying the inequality
\Eq{E:MP}{
\min(x)\le M(x)\le\max(x)\text{ for all }x \in I^p.
} 
Property \eq{E:MP} is referred as a \emph{mean property}. If the inequalities in \eq{E:MP} are strict for every nonconstant vector $x$, then we say that a mean $M$ is \emph{strict}. Moreover, for such objects, we define natural properties like continuity, symmetry (when the value of mean does not depend on the order of its arguments), monotonicity (which states that $M$ is nondecreasing in each of its variables), etc. 

A mean-type mapping is a selfmapping of $I^p$ which has a $p$-variable mean on each of its coordinates. More precisely, $\MM \colon I^p \to I^p$ is called a \emph{mean-type mapping} if $\MM=(M_1,\dots,M_p)$ for some $p$-variable means $M_1,\dots,M_p$ on $I$. In this framework, a function $K \colon I^p\to \R$ is called \emph{$\MM$-invariant} if it solves the functional equation $K \circ \MM=K$. Usually, we restrict solutions of this equation to the family of means and say about \emph{$\MM$-invariant means}.
Several authors studied invariant means during years, let us just mention the book by Borwein--Borwein \cite{BorBor87} and a comprehensive survey paper Jarczyk--Jarczyk \cite{JarJar18} and the references therein.

For a given $d,p \in \N$, a sequence 
\[
\alpha:=(\alpha_1,\dots,\alpha_d) \in \{1,\dots,p\}^d,
\]
and a $d$-variable mean $M \colon I^d\to I$ we define the mean $M^{(p;\alpha)}\colon I^p\to I$
by \Eq{E:parmean}{
M^{(p;\alpha)}(x_1,\dots,x_p):=M(x_{\alpha_1},\dots,x_{\alpha_d}) \text{ for all }(x_1,\dots,x_p)\in I^p.
}
\begin{exa}
Let $d=2$ and  $\A\colon \R^2\to \R$ be the bivariate arithemetic mean, $p \ge 3$ and $\alpha=(2,3)$ then $\A^{(p;\alpha)}\colon I^p \to I$ is given by 
\Eq{*}{
\A^{(p;\alpha)}(x_1,\dots,x_p)=\A^{(p;2,3)}(x_1,\dots,x_p)=\tfrac{x_2+x_3}2\text{ for all }(x_1,\dots,x_p)\in I^p.
}
\end{exa}

For the sake of completeness, let us introduce formally $\N:=\{1,\dots\}$, and $\N_p:=\{1,\dots,p\}$ (where $p\in \N$). Then, for $p \in \N$ and a vector $\dd=(d_1,\dots,d_p)\in \N^p$, let $\N_p^\dd:=\N_p^{d_1}\times\dots\times \N_p^{d_p}$.
Using this notations, a sequence of means $\MM=(M_1,\dots,M_p)$ is called \emph{$\dd$-averaging} mapping on~$I$ if each $M_i$ is a $d_i$-variable mean on $I$. 

For a $\dd$-averaging  mapping (for an example see Example \ref{exa:Pas23b-5}) $\MM$ and a vector of indexes $\alpha =(\alpha_1,\dots,\alpha_p)\in \N_p^{d_1}\times\dots\times \N_p^{d_p}=\N_p^\dd$ define a mean-type mapping $\MM_\alpha \colon I^p \to I^p$ by 
\Eq{*}{
\MM_\alpha:=\Big(M_1^{(p;\alpha_1)},\dots,M_p^{(p;\alpha_p)}\Big);
}
recall that $M_i^{(p,\alpha_i)}$-s were defined in \eq{E:parmean}. In the more explicit form we have
\Eq{*}{
\MM_\alpha(x_1,\dots,x_p)&=\Big(M_i^{(p,\alpha_i)}(x_1\dots,x_p)\Big)_{i=1}^p\\
&=\Big(M_i\big(x_{\alpha_{i,1}},\dots,x_{\alpha_{i,d_i}}\big)\Big)_{i=1}^p\\
&=\Big(M_1\big(x_{\alpha_{1,1}},\dots,x_{\alpha_{1,d_1}}\big),\dots,M_p\big(x_{\alpha_{p,1}},\dots,x_{\alpha_{p,d_p}}\big)\Big).
}

For a given $p\in\N$, $\dd=(d_1,\dots,d_p)\in \N^p$, and $\alpha \in \N_p^\dd$, we define the \emph{$\alpha$-incidence} graph $G_\alpha=(V_\alpha,E_\alpha)$ as follows: $V_\alpha:=\N_p$ and $E_\alpha:=\{(\alpha_{i,j},i) \colon i \in \N_p \text{ and }j \in \N_{d_i}\}$.

For the readers convenience and for the better understandability of the paper we recall two results from \cite{Pas23b} and \cite{MatPas21}, which will used later.

\begin{thm}[\!\!\cite{Pas23b}, Theorem~2 (a)-(d)]\label{thm:2}
Let $I \subset \R$ be an interval, $p \in \N$, $\dd \in \N^p$, $\alpha \in \N_p^\dd$, and $\MM=(M_1,\dots,M_p)$ be a $\dd$-averaging mapping on $I$. 
Assume that $G_\alpha$ is an ergodic graph, and $M_i$-s are continuous and strict for all $i \in \{1,\dots,p\}$.

There exists the unique, continuous, and strict \mbox{$\MM_\alpha$-invariant} mean $K_\alpha \colon I^p \to I$.
Moreover $\lim\limits_{n\to\infty}\MM_\alpha^n=\KK_\alpha$, where 
$\KK_\alpha \colon I^p \to I^p$ is defined as $\KK_\alpha=(K_\alpha,\dots,K_\alpha)$.
\end{thm}
 
\begin{prop}[Invariance principle; \cite{MatPas21}, Theorem~1] \label{prop:InvPrinc}
Let $\MM\colon I^p\to I^p$
be a mean-type mapping and $K:I^{p}\rightarrow I$ be an arbitrary mean. $K$
is a unique $\mathbf{M}$-invariant mean if and only if the sequence of
iterates $\left( \mathbf{M}^{n}\right) _{n\in \mathbb{N}}$ of the mean-type
mapping $\mathbf{M}$ converges to $\mathbf{K}:=\left( K,\dots ,K\right) $
pointwise on $I^{p}$.
\end{prop}

\section{Main result}
We aim to generalize Theorem~\ref{thm:2} in the following way
\begin{thm}\label{thm:main}
    Let $I \subset \R$ be an interval, $p \in \N$, $\dd \in \N^p$, $\alpha \in \N_p^\dd$, and $\MM=(M_1,\dots,M_p)$ be a $\dd$-averaging mapping on $I$. Assume that $\root(G_\alpha)$ is an ergodic graph, and $M_i$-s are continuous and strict for all $i \in \{1,\dots,p\}$. Then, there exists a unique and continuous \mbox{$\MM_\alpha$-invariant} mean $K_\alpha \colon I^p \to I$ such that
    \begin{equation}\label{3.D}
\lim_{n\to\infty}\MM_\alpha^n=\KK_\alpha,\qquad\mbox{(existence of a common narrative)}
    \end{equation}
    where 
    \[
    \KK_\alpha \colon I^p \to I^p,\qquad \KK_\alpha=(K_\alpha,\dots,K_\alpha),
    \]
    which depends on the root elements only. That is to say, there exists a mean $K_\alpha^* \colon I^{|R(G_\alpha)|} \to I$ such that 
    \begin{eqnarray*}
         K_\alpha(x_1,\dots,x_p)=K_\alpha^*(x_i \colon i \in R(G_\alpha)).\\
         \mbox{(the common narrative depends on only the opinion of the root)}
    \end{eqnarray*} 
\end{thm}
\begin{proof}
Let us assume without loss of generality that $R(G_\alpha)=(1,\dots,q)$ for some $q \in\{1,\dots,p\}$. 

If $q=p$ then all vertexes of $G_\alpha$ belong to the root. Whence $G_\alpha$ is irreducible, and $G_\alpha=\root(G_\alpha)$ is aperiodic. So,$G_\alpha$ is ergodic and this theorem is implied by Theorem~\ref{thm:2}. For the remaining part of the proof we assume that $q\in\{1,\dots,p-1\}$.

For $v \in V$ define $\rank(v)$ as the distance of $v$ from the closest vertex in $R(G_\alpha)$. If $v \in R(G_\alpha)$ then we set $\rank(v):=0$. 

For $k \ge 0$ define
$V_k:=\{v \in V \colon \rank(v)\le k\}$.
Obviously, $R(G_\alpha)=V_0 \subseteq V_1 \subseteq V_2 \subseteq \cdots$ and there exists $k_0$ such that $V=V_{k_0}$.

\noindent {\sc Step 1 (Means with coordinates in $V_0$).}
Since $V_0=(1,\dots,q)$ is the root of $G_\alpha$, we obtain that all means $[\MM_\alpha]_1,\dots,[\MM_\alpha]_q$ depends on the first $q$ variables only. 
Therefore let $\pi \colon I^p \to I^q$ be the projection to the first $q$ variables. 

Thus, if we define $\MM^*=(M_1,\dots,M_q)$ and $\alpha^*=(\alpha_1,\dots,\alpha_q)$ we get
\Eq{*}{
[\MM_\alpha]_s (x)&=[\MM_\alpha]_s (x_1,\dots,x_p)=M_s^{(p;\alpha_s)}(x_1,\dots,x_p)=M_s^{(q;\alpha_s)}(x_1,\dots,x_q)\\
&=[\MM^*_{\alpha^*}]_s(x_1,\dots,x_q)=[\MM^*_{\alpha^*}]_s \circ \pi (x_1,\dots,x_p)=[\MM^*_{\alpha^*}]_s \circ \pi (x)
}
for all $s \in V_0$ and $x=(x_1,\dots,x_p) \in I^p$. If we apply this equality to all admissible $s$ we get 
$\pi \circ \MM_\alpha=\MM^*_{\alpha^*}\circ \pi$. This, by easy induction, yields
\Eq{177}{
\pi \circ \MM_\alpha^n=(\MM^*_{\alpha^*})^n\circ \pi \text{ for all }n \in \N. 
}
However $G_{\alpha^*}$ is a graph $G_\alpha$ restricted to $V_0$, whence we obtain $G_{\alpha^*}=\root(G_\alpha)$.
Since $\root(G_\alpha)$ is ergodic, by Theorem~\ref{thm:2}, there exists the unique $\MM^*_{\alpha^*}$-invariant mean $K \colon I^q \to I$ and the sequence of iterates $((\MM^*_{\alpha^*})^n)_{n=1}^\infty$ converges to $\KK^*:=(K,\dots,K) \colon I^q \to I^q$. Then, by \eq{177}, $(\pi \circ \MM_\alpha^n)_{n=1}^\infty$ converges to $\KK^* \circ \pi$. In other words
\Eq{186}{
    ([\MM_\alpha^n]_i)_{n=1}^\infty \text{ converges to }K\circ \pi \text{ on }I^p \text{ for all }i \in V_0.
}

\medskip 
\noindent {\sc Step 2 (general case).}
Take $x \in I^p$ arbitrary and set $u_i:= \limsup_{n \to \infty} [\MM_\alpha^n]_i(x)$ ($i \in\{1,\dots,p\}$).
Property \eq{186} implies
\Eq{194}{
u_i=K \circ \pi(x) \text{ for all }i \in V_0.
}

Let $i_0 \in \{1,\dots p\}$ be a number that satisfies $u_{i_0}=\max\{u_i \colon i \in \{1,\dots,p\}\}$ with the minimal rank. We show that $\rank(i_0)=0$.

\medskip 
\noindent {\sc Step 2.1 $\rank(i_0)=0$.}
Assume to the contrary that $k:=\rank(i_0)>0$. Then $i_0 \in V_k \setminus V_{k-1}$ and, since $i_0$ have a minimal rank, we get $\rho:=\max\{u_i \colon i \in V_{k-1}\}<u_{i_0}$.
Whence for all $\varepsilon\in(0,+\infty)$ there exists $n_\varepsilon$ such that
\Eq{*}{
[\MM_\alpha^n]_i(x) \le \rho+\varepsilon \text{ for all }n\ge n_\varepsilon \text{ and } i \in V_{k-1}.
}
Then we have that 
\Eq{*}{
[\MM_\alpha^n]_i(x) \in [\min x,\rho+\varepsilon] \cap I=:A_\varepsilon \text{ for all }n\ge n_\varepsilon \text{ and } i \in V_{k-1}.
}
Moreover, there exists $m_\varepsilon$ such that
\Eq{*}{
[\MM_\alpha^n]_i(x) \in [\min x,u_i+\varepsilon] \cap I \subset [\min x,u_{i_0}+\varepsilon] \cap I =:B_\varepsilon 
}
 for all $n\ge m_\varepsilon$ and  $i \in V.$
 
Clearly $A_\varepsilon \subseteq B_\varepsilon$ for all $\varepsilon>0$.  
Now for $\varepsilon \ge 0$, let us define the set $\Lambda_\varepsilon :=\prod_{i =1}^p H_\varepsilon(i) \subset I^p$, where
\Eq{*}{
H_\varepsilon(i)=\begin{cases}
    A_\varepsilon &\text{ for }i\in V_{k-1};\\
    B_\varepsilon &\text{ for }i \in V \setminus V_{k-1}.
\end{cases}
}
Then $\MM_\alpha^n(x) \in \Lambda_\varepsilon$ for all $n \ge \max(n_\varepsilon,m_\varepsilon)$.
Moreover for all $i \in \{1,\dots,p\}$ the mapping $[0,+\infty) \ni \varepsilon \mapsto H_\varepsilon(i)$ are topologically continuous. Thus, so is $[0,+\infty) \ni \varepsilon \mapsto \Lambda_\varepsilon$. Therefore, the function
\Eq{*}{
\varphi \colon [0,+\infty) \ni \varepsilon \mapsto \sup \big\{[\MM_\alpha]_{i_0}(y) \colon y \in \Lambda_\varepsilon\big\}
\in [\min(x),\infty)
}
is also continuous. But, since $\rank(i_0)=k$, there exists $j \in V_{k-1}$ such that $(j,i_0) \in E$. Equivalently, the mean $[\MM_\alpha]_{i_0}$ depends on the $j$-th variable, say $\alpha_{i_0,q}=j$ for some $q \in \{1,\dots,d_{i_0}\}$.

Therefore, for all $\varepsilon>0$, we have
\Eq{*}{
\varphi(\varepsilon)&= \sup \big\{[\MM_\alpha]_{i_0}(y) \colon y \in \Lambda_\varepsilon\big\} \\
&= \sup \big\{M_{i_0}(y_{\alpha_{i_0,1}},\dots,y_{\alpha_{i_0,d_{i_0}}}) \colon (y_1,\dots,y_p) \in \Lambda_\varepsilon\big\} \\
&= \sup \big\{M_{i_0}(y_{\alpha_{i_0,1}},\dots,y_{\alpha_{i_0,d_{i_0}}}) \colon y_1 \in H_\varepsilon(1),\dots,y_p \in H_\varepsilon(p)\big\} \\
&\le \sup \big\{M_{i_0}(y_{\alpha_{i_0,1}},\dots,y_{\alpha_{i_0,d_{i_0}}}) \colon y_j \in A_\varepsilon,\text{ and } y_i \in B_\varepsilon \text{ for }i \ne j\big\} \\
&\le \sup \big\{M_{i_0}(z_1,z_2,\dots,z_{d_{i_0}}) \colon z_q \in A_\varepsilon,\text{ and } z_k \in B_\varepsilon\text{ for }k \ne q\big\}\\
&= \sup \big\{M_{i_0}(z) \colon z \in B_\varepsilon^{d_{i_0}}, z_q \in A_\varepsilon\big\}
=:\psi(\varepsilon).
}

However, since it is a supremum of a continuous function over a compact set, it attaches its maximum. Thus, for all $\varepsilon>0$, there exists $z^{(\varepsilon)} \in C_\varepsilon:=\{ z\in B_\varepsilon^{d_{i_0}} \colon z_q \in A_\varepsilon\}$ such that $\psi(\varepsilon)=M_{i_0}(z^{(\varepsilon)})$.
Since $M_{i_0}$ is continuous, we obtain that $\psi$ is nondecreasing and continuous.

Let $\bar z$ be any accumulation point of the set $\{z^{(1/n)} \colon n \in \N\}$. Clearly $\bar z$ belongs to the topological limit of $C_\varepsilon$, that is $\bar z \in \{ z \in [\min x,u_{i_0}]^{d_{i_0}} \colon z_q \in [\min x,\rho]\}$.

Since $M_{i_0}$ is a strict mean and $\rho<u_{i_0}$, we get $M_{i_0}(\bar z)<u_{i_0}$. Whence, since $\varphi$ and $\psi$ are nonincreasing and $\varphi\le \psi$ we get
\Eq{*}{
\lim_{\varepsilon \to 0^+}\varphi(\varepsilon) \le \lim_{\varepsilon \to 0^+}\psi(\varepsilon)
=\liminf_{n \to \infty}\psi(\tfrac1n)
=\liminf_{n \to \infty}M_{i_0}(z^{(1/n)})\le M_{i_0}(\bar z)<u_{i_0}.
}
Consequently, there exists $\varepsilon_0$ such that $\varphi(\varepsilon_0)<u_{i_0}$. Then, for all $n \ge \max(n_{\varepsilon_0},m_{\varepsilon_0})$ we have
$\MM_\alpha^n(x) \in \Lambda_{\varepsilon_0}$, that is 
$[\MM_\alpha^n]_{i_0}(x)\le \varphi(\varepsilon_0)$.
Therefore
\Eq{*}{
\limsup_{n \to \infty} [\MM_\alpha^n]_{i_0}(x)\le \varphi(\varepsilon_0)< u_{i_0}=\limsup_{n \to \infty} [\MM_\alpha^n]_{i_0}(x),
}
a contradiction. Thus $\rank(i_0)=0$. 
\medskip 

\noindent {\sc Step 2.2 (conclusion).}
Since $\rank(i_0)=0$ we have $i_0 \in V_0$. Whence,
by \eq{194}, we get 
\Eq{*}{
\limsup_{n \to \infty} [\MM_\alpha^n]_i(x)=u_i \le u_{i_0}=K\circ \pi(x) \text{ for any } i \in V.
}
Analogously, we can show the property
\Eq{*}{
\liminf_{n \to \infty} [\MM_\alpha^n]_i(x) \ge K\circ \pi(x) \text{ for any } i \in V,
}
Whence \eq{3.D} holds with $K_\alpha=K \circ \pi(x)$. If we set $K_\alpha^*:=K$ then, for all $i \in V$ and $x \in I^p$, we have 
\Eq{*}{
\lim_{n \to \infty} [\MM_\alpha^n]_i(x)=K\circ \pi(x)=K_\alpha^* \circ \pi(x)=}
\Eq{*}{=K_\alpha^*(x_i \colon i \in R(G_\alpha))=K_\alpha(x_1,\dots,x_p),
}
which completes the proof.
\end{proof}
\begin{cor}\label{cor:main}
    Let $I \subset \R$ be an interval, $p \in \N$, $\dd \in \N^p$, $\alpha \in \N_p^\dd$, and $\MM=(M_1,\dots,M_p)$ be a $\dd$-averaging mapping on $I$. Assume that $\root(G_\alpha)$ is an ergodic graph, and $M_i$-s are continuous and strict for all $i \in \{1,\dots,p\}$. Define $K_\alpha$ and $\KK_\alpha$ according to Theorem \ref{thm:main}. Then
\begin{enumerate}[{\rm (a)}]
     
    \item \label{2.E} $\KK_\alpha \colon I^p \to I^p$ is $\MM_\alpha$-invariant, that is $\KK_\alpha =\KK_\alpha\circ \MM_\alpha$;
    \item \label{2.F} if $M_1,\dots,M_p$ are nondecreasing with respect to each variable, then so is $K_\alpha$;
    \item \label{2.G} if $I=(0,+\infty)$ and $M_1,\dots,M_p$ are positively homogeneous, then every iterate of $\MM_\alpha$ and $K_\alpha$ are positively homogeneous.
\end{enumerate}
\end{cor}
\begin{proof}
Applying Theorem~\ref{thm:main} twice, for all $x \in I^p$ we have
\Eq{*}{
\KK_\alpha(x)=\lim_{n \to \infty} \MM_\alpha^n (x)=\lim_{n \to \infty} \MM_\alpha^n \big(\MM_\alpha(x)\big)=\KK_\alpha \circ \MM_\alpha(x),
}
which yields \eq{2.E}. 

Properties \eq{2.F} and \eq{2.G} are consequences of Theorem~\ref{thm:main} too. Indeed, if all $M_i$-s are nondecreasing (resp. homogenous) then so are all entries in $\MM_\alpha$. Then all entries in the sequence of iterates $\MM_\alpha^n$ also possess this property. Since it is inherited by the limit procedure, in view of \eq{3.D} we obtain that $K$ is nondecreasing (resp. homogenous). 
\end{proof}

So, we have a complete description if the root is ergodic, an immediate question is implied by this situation: Does something similar is true if the root is not connected? In other words, it has more than one component, more precisely, what happens if the root is not ergodic.

The following theorem says, that the nice characterization (see Theorem \ref{thm:main}) is available if and only if the root is ergodic.
\begin{thm}\label{thm:main2}
        Let $I \subset \R$ be an interval, $p \in \N$, $\dd \in \N^p$, $\alpha \in \N_p^\dd$, and $\MM=(M_1,\dots,M_p)$ be a $\dd$-averaging mapping on $I$ such that all $M_i$-s are continuous and strict. Then there exists the unique $\MM_\alpha$-invariant mean if and only if $\root(G_\alpha)$ is ergodic.
\end{thm}
\begin{proof}
    If $\root(G_\alpha)$ is ergodic then, as an immediate consequence of Theorem~\ref{thm:main}, we obtain that $\MM_\alpha$-invariant mean is uniquely determined.

    For the converse implication, let us take $p \in \N$, $\dd \in \N^p$, and $\alpha \in \N_p^\dd$ so that $\root(G_\alpha)$ is not connected or periodic. Moreover, let $\MM=(M_1,\dots,M_p)$ be an arbitrary $\dd$-averaging mapping on $I$ such that all $M_i$-s are continuous and strict.  This splits our proof into two parts. 

    {\sc Case 1.} If $\root(G_\alpha)$ is not connected then for all $v \in V$ there exists $\bar v \in V$ such that there is no path from $v$ to $\bar v$.
    Let us define, for all $v \in V$, sets 
    \Eq{*}{
    \Succ(v):=\{v\}\cup \{w \in V \colon v \leadsto w\}; \quad \Prec(v):=\{v\} \cup \{w \in V\colon w \leadsto v\}.
    }
    Clearly, for every $v \in V$ we have $\Succ(v) \cap \Prec(\bar v)=\emptyset$. Moreover $\Prec(v) \ne \emptyset$ for all $v \in V$.
    
    Moreover, each vertex has an in-neighbor and
    \Eq{E:NG-}{
    \Succ(w) \supseteq \Succ(v) \text{ for all }w \in N_G^-(v).
    }

    Now let $V_0$ be the maximal element of $\{\Succ(v) \colon v \in V\}$. 
    Then, in view of \eq{E:NG-} we have $\Succ(w)=V_0$ for all $w \in N_G^-(v)$. This implies that there are no edge from $V_0$ to $V \setminus V_0$ (that is, $E \cap (V_0 \times (V \setminus V_0))=\emptyset$).

    Moreover, by simple induction, we have $\Succ(w)=V_0$ for all $w \in \Prec(v_0)$. Since $v \in \Succ(v)$ for all $v\in V$ we get $\Prec(v) \subset V_0$ for all $v \in V_0$ (that is, $E \cap ((V \setminus V_0) \times V_0) =\emptyset$).

    Finally, we have $E \subset V_0^2 \cup (V \setminus V_0)^2$. Therefore every vector $x \in I^p$ of the form 
    \Eq{*}{
    x_i = \begin{cases}
        \gamma & \text{ if } i \in V_0\\
         \delta & \text{ if } i \in V \setminus V_0
    \end{cases}
    }
    (where $\gamma,\delta \in I$) is a fixed point of $\MM_\alpha$. By Proposition~\ref{prop:InvPrinc} we obtain that $\MM_\alpha$-invariant mean is not unique.

    {\sc Case 2.} If $\root(G_\alpha)=(V_0,E_0)$ is nonempty and periodic then there exists $c\ge 2$ and a partition $W_0,\dots W_{c-1}$ of $V_0$ such that 
    $E_0 \subseteq \bigcup_{i=0}^{c-1} W_i \times W_{i+1}$ (we set $W_{c+i}:=W_i$ for all $i \in \Z$). Take $\gamma,\delta \in I$ with $\gamma \ne \delta$ and define $x \in I^p$ as follows
    \Eq{*}{
    x_i=\begin{cases} 
    \gamma & \text{ if } i \in W_0, \\
    \delta  & \text{ if } i \in V_0 \setminus W_0.
    \end{cases}
    }
    But for all $i \in W_k$ means $[\MM_\alpha]_i$ depends only on arguments with indexes $W_{k-1}$. By the simple introduction, for all $n \in \N$ we get
    \Eq{*}{
    [\MM_\alpha^n]_i (x)=\begin{cases} 
    \gamma & \text{ if } i \in W_n,\\
    \delta  & \text{ if } i \in V_0 \setminus W_n.
    \end{cases}
    }
    Whence
    \Eq{*}{
    \lim_{n \to \infty} \max_{i \in \{1,\dots,p\}} [\MM_\alpha^n (x)]_i=\max(\gamma,\delta);\quad
    \lim_{n \to \infty} \min_{i \in \{1,\dots,p\}} [\MM_\alpha^n (x)]_i=\min(\gamma,\delta).
    }
    By Proposition~\ref{prop:InvPrinc}, this yields that the $\MM_\alpha$-invariant mean is not uniquely determined.
\end{proof}

\section{Further examples and discussion}
In the previous paper \cite{Pas23b} it was proved that the $\MM_\alpha$-invaraint mean is uniquely determined whenever each coordinate of $\MM$ is a continuous, strict mean and $G_\alpha$ is an ergodic graph. Now we improved this statement to the case when $\root(G_\alpha)$ is ergodic (Theorem~\ref{thm:main}). Clearly, this generalizes the previous setup since the root of an irreducible graph contains all vertexes. We was also able to show some related properties of this invariant mean (Corollary~\ref{cor:main}). It is also worth mentioning that the ergodicity of the root is unavoidable for the uniqueness of invariant mean (Theorem~\ref{thm:main2}). There are, however, still several interesting phenomena which we show in this section.

We start with the application of our main theorem. This example was already mentioned in \cite{Pas23b} and, in some sense, was the motivation for this research.

\begin{exa}[\!\!\cite{Pas23b}, Example~5]\label{exa:Pas23b-5}

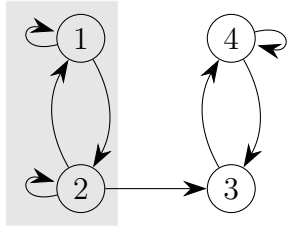
\begin{figure}[h]
\centering
\begin{tikzpicture}[>={Stealth[width=6pt,length=9pt]},
node distance=15mm,auto]
\filldraw [fill=black!10,draw=white] (-1,-0.5) rectangle (0.5,2.5);
\node[state,inner sep=3pt,minimum size=1pt] (b) at (0,0) {$2$};
\node[state,inner sep=3pt,minimum size=1pt] (c) at (2,0) {$3$};
\node[state,inner sep=3pt,minimum size=1pt] (d) at (2,2) {$4$};
\node[state,inner sep=3pt,minimum size=1pt] (a) at (0,2) {$1$};
\path[->] (a) edge [bend left] node {} (b);
\path[->] (a) edge [loop left] node {} (a);
\path[->] (b) edge [bend left] node {} (a);
\path[->] (b) edge [loop left] node {} (b);
\path[->] (b) edge node {} (c);
\path[->] (c) edge [bend left] node {} (d);
\path[->] (d) edge [bend left] node {} (c);
\path[->] (d) edge [loop right] node {} (d);
\end{tikzpicture}
\caption{Graph $G_\alpha$ related to Example~\ref{exa:Pas23b-5}.}
\end{figure}

Let $p=4$, $\dd=(2,2,2,2)$,
\Eq{*}{
\alpha=\big((1,2),(1,2),(2,4),(3,4)\big)\in \N_4^\dd 
\quad\text{ and }\quad
\MM=(\P_{-1},\P_1,\P_{-1},\P_1).
}
Then $\MM_\alpha$ is of the form
\Eq{*}{
\MM_\alpha(x,y,z,t)&=\bigg(\frac{2xy}{x+y},\frac{x+y}2,\frac{2yt}{y+t},\frac{z+t}2\:\bigg).
}

Clearly $\root(G_\alpha)$ is ergodic. Thus, by Theorem~\ref{thm:main}, we obtain that  there exists the unique $\MM_\alpha$-invariant mean $K_\alpha \colon\R_+^4 \to \R_+$, and it is of the form $K_\alpha(x,y,z,t)=K^*_\alpha(x,y)$, where $K^*_\alpha \colon \R_+^2 \to \R_+$. By $K_\alpha \circ \MM_\alpha=K_\alpha$ for all $x,y,z,t \in \R_+$ we obtain, 
\Eq{*}{
K_\alpha^*(x,y)=K_\alpha(x,y,z,t)=K_\alpha \circ \MM_\alpha (x,y,z,t)=}
\Eq{*}{
=K_\alpha \big(\tfrac{2xy}{x+y},\tfrac{x+y}2,\tfrac{2yt}{y+t},\tfrac{z+t}2\big)
=K_\alpha^* \big(\tfrac{2xy}{x+y},\tfrac{x+y}2\big).
}
Now we can use the folklore result stating that the arithmetic-harmonic mean is the geometric mean to obtain $K_\alpha^*(x,y)=\sqrt{xy}$ for $x,y \in \R_+$. Finally
$K_\alpha(x,y,z,t)=K^*_\alpha(x,y)=\sqrt{xy}$.

\end{exa}

Now we justify what happens if $\root(G_\alpha)$ is not connected. Then the iteration of elements in the root can be split into (at least two) independent iteration processes. There appears a natural problem: if convergence of elements in the root yields the convergence in the whole graph. 

In the next example, we show that this is not the case. This example is much different from the previous approaches. Namely, we are going to study the iterations (and invariant means) only for two vectors. Furthermore, in this example, the mean-type mapping contains a mean (denoted by $F$) which is not given explicitly.

\begin{exa}\label{exa:E1}
    Let $I \subset \R$ be an interval and $a,b,c,d \in I$ with $a<b<c<d$. There exists a symmetric, continuous, and strict mean $F \colon I^3 \to I$ such that $F(a,d,b)=c$ and $F(a,d,c)=b$.

    Set $\dd:=(1,1,3,3)$, $\dd$-averaging mapping $\MM:=(\id,\id,F,F)$ (here $\id \colon I \to I$ stands for the identical function) and set 
    \[
    \alpha:=\big((1),(2),(1,2,4),(1,2,3)\big) \in \N_4^{\dd}
    \]
    Then we have
    \Eq{*}{
    \MM_\alpha(x,y,z,t)=\big(x,y,F(x,y,t),F(x,y,z)\big).
    }
    In particular for $v_1:=(a,d,b,c)$ and $v_2:=(a,d,c,b)$ we have $\MM_\alpha(v_i)=v_{3-i}$ ($i \in \N_2$).

    Observe that $R(G_\alpha)=\{1,2\}$ and thus $[\MM_\alpha^n]_i$ is convergent for all $i \in R(G_\alpha)$, although it is (in general) not convergent for indexes which do not belong to the root.

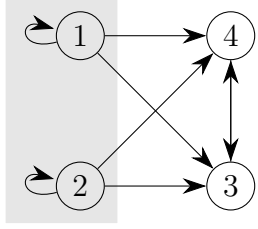
\begin{figure}[h]
\centering
\begin{tikzpicture}[>={Stealth[width=6pt,length=9pt]},
node distance=15mm,auto]
\filldraw [fill=black!10,draw=white] (-1,-0.5) rectangle (0.5,2.5);
\node[state,inner sep=3pt,minimum size=1pt] (b) at (0,0) {$2$};
\node[state,inner sep=3pt,minimum size=1pt] (c) at (2,0) {$3$};
\node[state,inner sep=3pt,minimum size=1pt] (d) at (2,2) {$4$};
\node[state,inner sep=3pt,minimum size=1pt] (a) at (0,2) {$1$};
\path[->] (a) edge [loop left] node {} (a);
\path[->] (b) edge [loop left] node {} (b);
\path[->] (a) edge node {} (c);
\path[->] (b) edge node {} (c);
\path[->] (d) edge node {} (c);
\path[->] (a) edge node {} (d);
\path[->] (b) edge node {} (d);
\path[->] (c) edge node {} (d);
\end{tikzpicture}
\caption{Graph $G_\alpha$ related to Example~\ref{exa:E1}.}
\end{figure}

Now we define means $L_i,U_i \colon I^4\to I$ ($i \in \N_4$) by
\Eq{*}{
L_i(v)=\liminf_{n \to \infty} [\MM_\alpha^n(v)]_i\qquad \text{ and }\qquad U_i(v)=\limsup_{n \to \infty} [\MM_\alpha^n(v)]_i.
}
Clearly, for all $i \in \N_4$, we have  $L_i\circ \MM_\alpha=L_i$ and $U_i\circ \MM_\alpha=U_i$, that is $L_i$-s and $U_i$-s are $\MM_\alpha$-invariant. For vectors $v_i$ ($i \in \N_2$) these means are 
\Eq{*}{
L_1(v_i)=U_1(v_i)=a,\hskip13.3mm&\quad L_2(v_i)=U_2(v_i)=d,\\
L_3(v_i)=L_4(v_i)=\min(b,c),&\quad U_3(v_i)=U_4(v_i)=\max(b,c).
}

\end{exa}

\subsection{Markov chains and weighted arithmetic means, when the root is not ergodic}

Let us assumed that the involved means are weighted arithmetic means (trivial weights, so projections are allowed). Then the iteration process is nothing else but a discrete Markov chain. Even, if the root is not ergodic, the limit of the iteration is unique, however there is no unique invariant mean in this case (see Theorem \ref{thm:main2}).

Observe that if all means are weighted arithmetic means then the iteration process can be completely characterized by a right stochastic matrix $A$ which describes the share of the individual opinion in the next step.

Let's consider the following numerical example.
\begin{exa}\label{exa:WQA}
    Let $d=4$ and the $d$-averaging mapping $M\colon\R^4\to\R^4$ given by
    \[
M(x_1,x_2,x_3,x_4)=\left(x_1,x_2,\frac{x_1+2x_2+3x_3+3x_4}{9},\frac{2x_1+x_2+x_3+2x_4}{6}\right).
    \]
    Then the corresponding discrete Markov chain can be presented as the matrix
    \[
    A:=\begin{bmatrix}
    1&0&0&0\\
    0&1&0&0\\
    \frac{1}{9}&\frac{2}{9}&\frac{3}{9}&\frac{3}{9}\\[2mm]
    \frac{2}{6}&\frac{1}{6}&\frac{1}{6}&\frac{2}{6}
    \end{bmatrix}.
    \]
    Then we can see that $M(x)=Ax$ for all $x \in \R^4$. 
    The limit of the iteration process is
    \[
    \lim\limits_{n\to\infty}A^n=\begin{bmatrix}
        1&0&0&0\\
    0&1&0&0\\
    \frac{10}{21}&\frac{11}{21}&0&0\\[2mm]
    \frac{13}{21}&\frac{8}{21}&0&0
    \end{bmatrix}
    \]
    This means that the opinion of the users will be the mixtures of the two components of the root with weights $\tfrac{10}{21},\tfrac{11}{21}$ and $\tfrac{13}{21},\tfrac{8}{21}$ respectively.

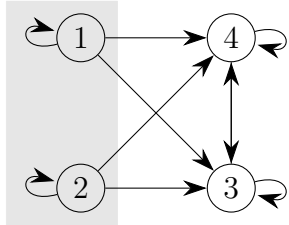
\begin{figure}[h]
\centering
\begin{tikzpicture}[>={Stealth[width=6pt,length=9pt]},
node distance=15mm,auto]
\filldraw [fill=black!10,draw=white] (-1,-0.5) rectangle (0.5,2.5);
\node[state,inner sep=3pt,minimum size=1pt] (b) at (0,0) {$2$};
\node[state,inner sep=3pt,minimum size=1pt] (c) at (2,0) {$3$};
\node[state,inner sep=3pt,minimum size=1pt] (d) at (2,2) {$4$};
\node[state,inner sep=3pt,minimum size=1pt] (a) at (0,2) {$1$};
\path[->] (a) edge [loop left] node {} (a);
\path[->] (b) edge [loop left] node {} (b);
\path[->] (c) edge [loop right] node {} (c);
\path[->] (d) edge [loop right] node {} (d);
\path[->] (a) edge node {} (c);
\path[->] (b) edge node {} (c);
\path[->] (d) edge node {} (c);
\path[->] (a) edge node {} (d);
\path[->] (b) edge node {} (d);
\path[->] (c) edge node {} (d);
\end{tikzpicture}
\caption{Graph $G_\alpha$ related to Example~\ref{exa:E1}.}
\end{figure}

\end{exa}

\section{Conclusion and further research}
We gave a new model of information spreading on networks. The key concept in our investigation was the notion of invariant means of certain averaging mappings. The roots in the network has a special role as we have proved. Actually, the accepted narrative of the network depends on only the opinion of root elements.

According to our best knowledge this model of spreading information on networks and information bubbles is brand new in the literature. So, several open problems can be posed concerning this new approach. 

Let us mention just a few. One of the most important questions in our opinion is the better understanding of the case when the root is not ergodic. There is no unique invariant mean in this case (see Theorem~\ref{thm:main2}), however the iteration process results a unique limit, which can be considered as a narrative at the end. It is not clear what is the role of the other invariant means in this case.

A good start for the investigation of this would be $d$-averaging mappings containing only weighted arithmetic means (see Example \ref{exa:WQA}). 

Another important question, which can simplify further investigations, is the following. If the root contains $k$ different components with $\alpha_i,\ i=1,\dots,k$ variables. And the corresponding invariant means are $K_1,\dots, K_k$, then do we get the same situation or not, if we substitute the aggregation process of root elements with the corresponding invariant mean at the very beginning?

Furthermore, based on Example~\ref{exa:E1}, we know that the convergence on sequence of iterates on the root set (in general) does not imply that it is convergent on remaining elements. On the other hand, we conjecture that it would be the case under some additional assumptions. For example, if we additionally assume that all means are monotone in their parameters.

Another possible direction to make the model more realistic is to assume that the users and the influencers change their aggregation process in time. A possible approach to grab this is to use random means (defined in \cite{BarBur22a}).

Finally, it is not known how a modification of the mean in one vertex impacts to whole iteration process. More precisely, is it true that if a single individual slightly changes the way of aggregating the information then it will not have a big impact to the remaining part of the graph.
\section{Acknowledgement}

P. Burai acknowledges the support of the Hungarian National Research Development and Innovation Office (NKFIH) through the grant TKP2021-NVA-02.


\end{document}